\documentclass[lettersize,journal]{IEEEtran}
\usepackage{newtxtext,newtxmath}

\usepackage{amsmath,amssymb,amsfonts,amsthm}
\usepackage{mathrsfs}
\usepackage{mathtools}
\usepackage{array}
\usepackage{tabularx}
\usepackage{multirow}
\usepackage{adjustbox}
\usepackage{booktabs}
\usepackage{color,colortbl}
\definecolor{lavender}{rgb}{0.9,0.9,0.98}
\usepackage{graphicx}
\usepackage{subcaption}
\usepackage{asymptote}
\usepackage{float}
\usepackage{algpseudocode}
\usepackage[linesnumbered,ruled,vlined]{algorithm2e}
\SetKwInput{KwInput}{Input}
\SetKwInput{KwOutput}{Output}
\SetKwInput{KwInit}{Initialization}
\usepackage{tikz}
\usetikzlibrary{
  decorations.pathreplacing,
  shapes,
  arrows,
  shadows,
  patterns,
  shapes.geometric,
  decorations.markings
}
\tikzstyle{every picture}+=[remember picture]
\usepackage{pgfplots}
\pgfplotsset{major grid style={densely dotted,black!70}}
\pgfplotsset{compat=1.9}
\usepackage{cite}
\usepackage{lipsum}
\usepackage{soul}
\usepackage{textcomp}
\usepackage[hidelinks]{hyperref}
\usepackage[capitalise]{cleveref}
\usepackage{multicol}
\allowdisplaybreaks
\newtheorem{theorem}{Theorem}

\begin{document}

\title{Detection of Number of Subcarriers of OFDM Systems using Eigen-Spectral Analysis}

\author{Vishnu Priya Chekuru, Ganapathiraju S S Ananya Varma, Arti Yardi, and Praful Mankar
\thanks{The authors are with the Signal Processing and Communication Research Center (SPCRC), International Institute of Information Technology Hyderabad, India. Email: \{vishnu.chekuru,ananya.varma\}@students.iiit.ac.in and \{arti.yardi,praful.mankar\}@iiit.ac.in.}
}

\maketitle

\begin{abstract}
Orthogonal Frequency-Division Multiplexing (OFDM) is widely used in modern wireless communication systems due to its robustness against time-dispersive channels. 
In this work, we consider a \textit{non-cooperative} scenario where the receiver does not have prior knowledge of the OFDM parameters such as the number of subcarriers and the aim is to estimate them using the received data. Such a setup has applications in cognitive radio networks. 
For this blind OFDM parameter estimation problem, we provide a novel method based on eigen-spectral analysis of the covariance matrix corresponding to the received data. In particular, we show that the covariance matrix exhibits a distinctive rank property under correct segmentation of the received symbols, reflecting a characteristic behavior in its eigenvalue spectrum that facilitates accurate estimation of the number of subcarriers.
The proposed method is more general than existing approaches in the literature, as it can detect an arbitrary number of subcarriers and its performance remains independent of the modulation scheme. 
The numerical results show that the proposed method accurately detects the number of subcarriers with high probability even at low SNR. 
\end{abstract}


\section{Introduction}
\label{sec:introduction}

%
Orthogonal Frequency-Division Multiplexing (OFDM) has become the cornerstone of modern broadband communication systems, including LTE, Wi-Fi, and 5G, due to its high spectral efficiency and robustness against multipath fading. By decomposing a high-rate serial stream into multiple parallel low-rate subcarriers, OFDM effectively transforms a frequency-selective channel into a set of flat-fading subchannels, simplifying the equalization and synchronization at the receiver end~\cite{Goldsmith_Book}. 
Further in OFDM systems, typically a cyclic prefix (CP) is inserted in each OFDM symbol that helps to mitigate inter-symbol interference (ISI) that the wireless fading channel may introduce.

%
In the \textit{cooperative} scenario, the receiver knows the OFDM system configuration, such as the number of subcarriers, subcarrier frequencies, and CP length. Note that when the OFDM parameters are known, the received noise-affected data can be demodulated to recover the transmitted data~\cite{Goldsmith_Book}.
In this work, we consider a \textit{non-cooperative} scenario where the OFDM system configuration is not completely known to the receiver. In this case, in order to recover the transmitted data, one needs to first estimate the underlying OFDM parameters only using the received data and thus this problem is termed as \textit{blind estimation of OFDM parameters}~\cite{OFDM_2006_Compare, OFDM_para_esti_Globecomm_2007, Guard_length_estimate, OFDM_Punchihewa_Trans_Wireless_2011, Freq_domain_ODFM_esti_2013_EURASIP, RADAR_OFDM_Globecom_2020, Testbed_OFDM_esti_2021, GFDM_para_esti_TCOM_2016}.
Such scenarios have applications in cognitive radios where the secondary receiver may not know the system configuration of the primary receiver~\cite{Congnitive_radios}. 
%


%
This problem of blind estimation of OFDM parameters has been studied in \cite{OFDM_2006_Compare, OFDM_para_esti_Globecomm_2007, Guard_length_estimate, OFDM_Punchihewa_Trans_Wireless_2011, Freq_domain_ODFM_esti_2013_EURASIP, RADAR_OFDM_Globecom_2020, Testbed_OFDM_esti_2021, GFDM_para_esti_TCOM_2016, CNN_AMC_OFDM_VTC_2020, AMC_Blind_OFDM_Access_2021, TOR_GAN_Signal_reconstruction, LSTM_Signal_reconstruction}.
While \cite{OFDM_2006_Compare, OFDM_para_esti_Globecomm_2007, Guard_length_estimate, OFDM_Punchihewa_Trans_Wireless_2011, Freq_domain_ODFM_esti_2013_EURASIP, RADAR_OFDM_Globecom_2020} explore the underlying structural properties of the received data
for parameter estimation, 
\cite{CNN_AMC_OFDM_VTC_2020, AMC_Blind_OFDM_Access_2021, TOR_GAN_Signal_reconstruction, LSTM_Signal_reconstruction} use machine learning-based methods to identify the OFDM parameters.
Earlier works \cite{OFDM_2006_Compare, OFDM_para_esti_Globecomm_2007, Guard_length_estimate} explore the correlation and cyclostationarity properties~\cite{Cyclostationarity_OFDM_1999} of the received data (either in time-domain or frequency domain) to estimate the OFDM parameters and they provide methods for the wireless channel without fading. 
Authors of \cite{OFDM_Punchihewa_Trans_Wireless_2011, Freq_domain_ODFM_esti_2013_EURASIP} extend the ideas based on the second-order cyclostationarity properties to provide algorithms for the wireless channel with fading.
In \cite{RADAR_OFDM_Globecom_2020}, authors consider the covariance matrix corresponding to the received data and study the eigenvalues of this matrix to estimate the OFDM parameters. 
%
Authors of \cite{Testbed_OFDM_esti_2021} provide
the testbed implementations of OFDM parameter estimation.

While a variety of algorithms are available in the literature, most methods suffer from one or more limitations: (i) they need the perfect knowledge of the channel condition, (ii) algorithms are not robust for the multipath or block fading channels, (iii) probability of correct parameter estimation degrades for the higher-order modulation schemes (particularly for works exploiting cyclostationarity), (iv) algorithms are designed for the number of OFDM subcarriers equal to a power of two.
Further, restricting subcarriers to a power of two may limit the applicability, as  the broader subcarriers are usually zero-padded in practical scenario.
We further observe that, while most of the existing works focus on providing algorithms for the OFDM parameter estimation, they lack in adequate analytical justification of the correctness of their algorithms.
%
%

%
%
In this work, we aim to remedy the lacunae of the existing works and provide a novel method for this parameter estimation problem for OFDM system. The key contributions and features of our method are summarized next:
\begin{itemize}
%
\item \textit{Novel method:} We consider the covariance matrix corresponding to the received data and study its eigenvalues to estimate of the number of subcarriers. Towards this, we use an approach based on the minimum description length (MDL)~\cite{MDL_Old_1989}, which has been widely used to detect the breakpoint separating signal and noise eigenvalues. 
%
\item \textit{Analytical characterization:} We study the rank properties of the covariance matrix of the noise-free data and prove that this matrix will be rank deficient only for the correct parameters (see \cref{thm:rank_property}). This  enabled us to correctly identify the number of subcarriers using MDL (see Algorithm \ref{alg:subcarrier_detection}).
%
\item \textit{Advantages over existing approaches:} Crucial advantages of our method are: (i) it works for arbitrary number of subcarriers (it need not be a power of two, see Fig. \ref{fig:Pd_N}), (ii) it works for multipath block fading channels, (iii) probability of correct parameter estimation is not sensitive to the choice of modulation scheme (see Fig. \ref{fig:Pd_SNR_QAM}).
\end{itemize}

\textit{Organization:} 
The OFDM system and received signal models are presented in Sec.~\ref{sec:System_model}. The proposed algorithm is detailed in Sec.~\ref{sec:subcarrier_detection}. Simulation-based performance analysis is discussed in Sec.~\ref{sec:results}, with conclusion in Sec.~\ref{Section_conclusion}.

\textit{Notation:} Vectors are denoted by bold lowercase letters (e.g., $\mathbf{v}$) and matrices by bold uppercase letters (e.g., $\mathbf{M}$). 
A vector $\mathbf{v}$ is represented as a column vector. The transpose, conjugate, and conjugate transpose are represented by $(\cdot)^T$, $(\cdot)^*$, and $(\cdot)^H$, respectively. The operator ${\rm vec}(\mathbf{M})$ denotes vectorization of matrix $\mathbf{M}$. The $(p,q)$-th element of matrix $\mathbf{M}$ and the $p$-th element of vector $\mathbf{v}$ are written as $\mathbf{M}(p,q)$ and $\mathbf{v}(p)$, respectively. $\mathbf{I}_K$ denotes the $K\times K$ identity matrix, and $\mathbf{Q}_N$ represents the $N\times N$ discrete Fourier transform (DFT) matrix with entries $\mathbf{Q}_N(p,q)=e^{-j2\pi pq/N}$. The 
rank and determinant of $\mathbf{M}$ are denoted by
$\mathrm{rk}(\mathbf{M})$ and $|\mathbf{M}|$ respectively.

\section{System Model}
\label{sec:System_model}
This paper focuses on estimating the number of subcarriers in an OFDM system from a passively observed signal. The following subsections present the design of the OFDM transmitter and the received signal model used for the estimation process.

\subsection{OFDM Transmitter}
\label{sec:ODFM_transmitter}
We consider an orthogonal frequency division multiplexing (OFDM) transmitter comprising $N$ orthogonal subcarriers separated by a uniform frequency spacing $\Delta f$, as shown in Fig.~\ref{fig:Block_Diagram}. Let $\bar{\mathbf{X}}_k\in\mathbb{C}^{N\times M}$ be the $k$-th  block consisting of streams of QAM symbols of length $M$ that parallely modulates $N$ subcarriers. 

\begin{figure}[H]
    \centering
    \hspace{.6cm}
    \includegraphics[width=0.5\textwidth]{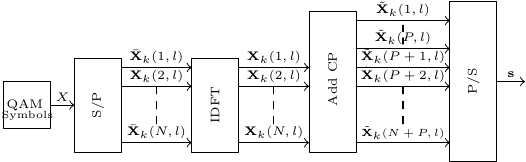}
\caption{Schematic Block Diagram of OFDM Transmitter.}
\label{fig:Block_Diagram}
\end{figure}

After inverse discrete Fourier transform (IDFT), the $k$-th time domain block can be written as
$$\mathbf{X}_k=\frac{1}{N}\mathbf{Q}_N^{-1}\bar{\mathbf{X}}_k,$$
where $\mathbf{Q}_N$ is $N\times N$ DFT matrix.
Next, to handle the inter-symbol interference (ISI) and aid demodulation with some useful properties, a  cyclic prefix (CP) of length $P$ is added to each OFDM symbol. Basically, CP appends the last $P$ rows of block $\mathbf{X}_k$ to its beginning, which yields $(N+P)\times M$ CP-OFDM block matrix as below
\begin{equation}
\tilde{\mathbf{X}}_k = 
\begin{bmatrix}
\mathbf{X}_k(N-P+1,1) & \dots & \mathbf{X}_k(N-P+1,M)\\
\vdots & \vdots & \vdots \\
\mathbf{X}_k(N,1) & \dots & \mathbf{X}_k(N,M)\\
\mathbf{X}_k(1,1) & \dots & \mathbf{X}_k(1,M)\\
\vdots & \vdots & \vdots \\
\mathbf{X}_k(N,1) & \dots & \mathbf{X}_k(N,M)
\end{bmatrix},\nonumber
\end{equation}
where $\mathbf{X}_k(i,j) $ represents the $(i,j)$-th entry of $\mathbf{X}_k$. 
Let $\tilde{\mathbf{x}}_{k,l}$ denotes the $l$-th column of $\tilde{\mathbf{X}}_k$ for $l = 1, 2, \ldots, M$. Thus, each column of $\tilde{\mathbf{X}}_k$ represents one CP-OFDM symbol of length $(N+P)$, and the overall block contains $M$ such symbols, giving a total serialized length $T = M(N+P)$. After inserting CP, the parallel streams of the CP-OFDM block are converted into a serial time-domain sequence for transmission. As parallel to serial conversion is performed at symbol by symbol basis, the transmission sequence $\mathbf{s}_k\in\mathbb{C}^{T\times 1}$ corresponding to $k$-th CP-OFDM block becomes
\begin{equation}
    \mathbf{s}_k = \mathrm{vec}(\tilde{\mathbf{X}}_k)=\begin{bmatrix}
        \tilde{\mathbf{x}}_{k,1}\\ \vdots \\\tilde{\mathbf{x}}_{k,M}
    \end{bmatrix},\label{eq:block_m_seq}
\end{equation}
where $\mathrm{vec}(\cdot)$ denotes the vectorization operator.

\subsection{Received Signal Model}
\label{sec:received_signal_model}
It is assumed that the receiver uses $K$ CP-OFDM blocks for estimating the number of subcarriers that are received via multipath channel while assuming the perfect synchronization with the transmitter. Let $\mathbf{s} =
[\mathbf{s}_1^T, \mathbf{s}_2^T,\cdots,\mathbf{s}_K^T]^T$ be the transmitted sequence of $K$ number of blocks. The multipath channel model is assumed to follow quasi block fading, wherein the $L$-tap channel response remains constant within each block and takes independent realizations across different blocks. The channel response under $k$-th block is denoted as  $\mathbf{h}_k = [h_k(0), h_k(1), \cdots, h_k(L-1)]^T$ such that $\mathbf{h}_k\sim\mathcal{CN}(\mathbf{0},\sigma_h^2\mathbf{I}_L)$. Due to the finite channel memory $(L{-}1)$, convolution of $\mathbf{s}$ with $\mathbf{h}_k$'s induces overlap between the adjacent blocks. Thus, in response to the transmitted sequence $\mathbf{s}$, the received sequence can be modeled as $\mathbf{r}=[\mathbf{r}_1^T, \mathbf{r}_2^T,\dots, \mathbf{r}_K^T]^T$ such that 
\begin{equation}
\mathbf{r}_k = \mathbf{H}_k\,\mathbf{s}_k + \mathbf{B}_k\,\mathbf{s}_{k-1} + \mathbf{w}_k,
\label{eq:block_model}
\end{equation}
where $\mathbf{r}_k \in \mathbb{C}^{T\times1}$, $\mathbf{w}_k \sim \mathcal{CN}(\mathbf{0}, \sigma_n^2 \mathbf{I}_{T})$, and for $k=1$, $\mathbf{H}_k =$

\begin{equation*}
\small
\begin{bmatrix}
h_k(0) & 0 & 0 &\cdots & \cdots &0 &0 \\
h_k(1) & h_k(0) & 0 &\cdots & \cdots & 0 & 0\\
h_k(2) & h_k(1) & h_k(0) &\cdots & \cdots & 0 & 0\\
\vdots & \ddots & \ddots & \ddots & \ddots& \vdots & \vdots \\
0&\dots&0&h_k(L-1) & \cdots & h_k(1) & h_k(0)
\end{bmatrix},
\label{eq:H_expanded}
\normalsize
\end{equation*}
is a $T\times T$ matrix that captures the intra-block interference, and 
\begin{equation*}
\small
\mathbf{B}_k =
\begin{bmatrix}
0 & 0 & \cdots & h_k(L{-}1) & \cdots& \cdots & h_k(1) \\
0 & 0 & \cdots & 0 & h_k(L{-}1) & \cdots&h_k(2) \\
\vdots & \ddots & \ddots & \vdots & \ddots & \vdots \\
0 & \cdots & 0 & 0 & \cdots &\cdots & h_k(L{-}1)\\
0 & \cdots & 0 & 0 & \cdots &\cdots & 0 \\
\vdots & \cdots & \vdots & \vdots & \vdots &\vdots & \vdots \\
0 & \cdots & 0 & 0 & \cdots &\cdots & 0
\end{bmatrix},
\label{eq:B_expanded}
\normalsize
\end{equation*}
is a  $T\times T$ matrix that captures the inter-block interference received from the $(k-1)$-th block. Thus, using \eqref{eq:block_model}, the received sequence can be modeled using lower bidiagonal channel matrix with Toeplitz blocks as 
\begin{equation}
\mathbf{r} =
\underbrace{
\begin{bmatrix}
\mathbf{H}_1 & \mathbf{0}   & \cdots & \mathbf{0} \\
\mathbf{B}_2 & \mathbf{H}_2 & \ddots & \vdots \\
\vdots       & \ddots       & \ddots & \mathbf{0} \\
\mathbf{0}   & \cdots       & \mathbf{B}_K & \mathbf{H}_K
\end{bmatrix}
}_{\triangleq\,\mathbf{H} }
\begin{bmatrix}
\mathbf{s}_1 \\[2pt]
\mathbf{s}_2 \\[2pt]
\vdots \\[2pt]
\mathbf{s}_K
\end{bmatrix}
+ \mathbf{w},
\label{eq:stacked_H}
\normalsize
\end{equation}
where $\mathbf{w}=[\mathbf{w}_1^T,\mathbf{w}_2^T,\dots,\mathbf{w}_K^T]^T$. In compact form, \eqref{eq:stacked_H} can be written as $\mathbf{r} = \mathbf{H}\mathbf{s} + \mathbf{w}$. Since each CP-OFDM symbol has length $(N+P)$, $\mathbf{H}_k$ can equivalently be viewed as a block-Toeplitz matrix comprising $M\times M$ sub-blocks, each of size $(N+P)\times(N+P)$.

\section{Estimation of Number of Subcarriers}
\label{sec:subcarrier_detection}
This section presents a method to estimate the number of subcarriers $N$ of CP-OFDM system assuming the CP length $P$ and number of channel taps $L$ (i.e. delay spread) are known. Note that CP length can be accurately estimated using a widely used cyclostationary  property of the autocorrelation function of the OFDM signal (for details, please refer to \cite{Cyclostationarity_OFDM_1999, OFDM_2006_Compare, OFDM_Punchihewa_Trans_Wireless_2011}). 
The proposed method is based on a key observation  that the CP introduces deterministic linear dependencies when the received signal is correctly segmented. This property allows to find the dimension of signal subspace only for a particular choice of segmentation, which in turn facilitates the estimation of $N$. 

The received sequence $\mathbf{r}$ of length $KT$ is first divided into $M^\prime=\left\lfloor\frac{KT}{N^\prime}\right\rfloor$ segments, each of length  $N^\prime$. Using these segments, the received sequence $\mathbf{r}$ is fragmented into a $N^\prime\times M^\prime$ matrix as below   
\begin{equation}
    \mathbf{R}_{\rm N^\prime} = [\tilde{\mathbf{r}}_1, \tilde{\mathbf{r}}_2, \dots, \tilde{\mathbf{r}}_{M^\prime}],\label{eq:R_N_prime}
\end{equation}
where $\tilde{\mathbf{r}}_k=[\mathbf{r}((k-1)N^\prime+1),\dots,\mathbf{r}(kN^\prime) ]^T$.



For the correct segmentation, we show that each column of $\mathbf{R}_{\rm N^\prime}$ contains exactly one full CP-OFDM symbol convolved with the channel response during the corresponding block. This allows us to relate the rank of $\mathbf{R}_{\rm N^\prime}$ (for the noise-free case) with the actual number of subcarriers $N$. This property holds only for $N^\prime=N^\star=N+P$; for other $N^\prime\neq N+P$, the matrix $\mathbf{R}_{\rm N^\prime}$ is a full rank matrix.
This property is formally stated in the following theorem.
\begin{theorem}
\label{thm:rank_property}
For the received noise-free sequence $\mathbf{r}$ with CP length $P\geq L$, the rank of $N^\prime\times M^\prime$ matrix  $\mathbf{R}_{\rm N^\prime}$, constructed with $M^\prime\geq N+P$, is
\begin{align}
\mathrm{rk} \big(\mathbf{R}_{\rm N^\prime}\big)
=\begin{cases}
    N + L - 1, & \text{if}~ N^{\prime}=N^\star=N+P,\\
    N^{\prime} & \text{otherwise}.
\end{cases}
\label{eq:rank_cases}
\end{align}
\end{theorem}
\begin{proof}
We first discuss the case of $N^\prime=N+P$. For this $N^\prime$, we have $M^\prime=\frac{KT}{N+P}=KM$. Recall from the input-output relation given in \eqref{eq:block_model},  $\mathbf{H}_k$ is the lower-triangular Toeplitz matrix and  $\mathbf{B}_k$ captures the spillover from block $\mathbf{s}_{k-1}$ into the head of block $\mathbf{s}_k$. 
Using \eqref{eq:block_model}, we can write the $i$-th element of vector $\tilde{\mathbf{r}}_l$ corresponding to the $l$-th CP-OFDM symbol of $k$-th block received under noise-free scenario as 
\begin{align*}
\tilde{\mathbf{r}}_l(i)&= \sum_{\ell=1}^{T}\mathbf{H}_k(i,\ell)\mathbf{s}_{k}(\ell) +\mathbf{B}_k(i,\ell)\mathbf{s}_{k-1}(\ell),
\end{align*}
where $l= (k-1)(N+P)+1,\dots,(k-1)(N+P)+M$. Further, using \eqref{eq:block_m_seq} and the structures of $\mathbf{H}_k$ and $\mathbf{B}_k$, we can write the $i$-th element of $\tilde{\mathbf{r}}_l$ as
\begin{align*}
\tilde{\mathbf{r}}_l(i)&= \sum_{\ell=0}^{\min(i,L)-1}\mathbf{h}_k(\ell)\tilde{\mathbf{x}}_{k,l}(i-\ell) + {\rm F},
\end{align*}
where the term ${\rm F}$ represents the IBI or ISI and is given by
{\small \begin{align*}
    {\rm F}=\begin{cases}
        0, &\text{if}~ i\geq L,\\        \sum\limits_{\ell=i}^{L-1}\mathbf{h}_k(\ell)\tilde{\mathbf{x}}_{k,l}(N+P+1-\ell), & \text{if}~ i<L, l>1, \\
        \sum\limits_{\ell=i}^{L-1}\mathbf{h}_k(\ell)\tilde{\mathbf{x}}_{k-1,M}(N+P+1-\ell), & \text{if}~ i<L, l=1.
    \end{cases}
\end{align*}}


Here, it is worth noting that the interference ${\rm F}$ from previous block or symbol is contributed to first  $L-1$ samples of vector $\tilde{\mathbf{r}}_l$. Further, for $i=L,\dots,P$, we can simplify    
\begin{align}
\tilde{\mathbf{r}}_l(i)&= \sum_{\ell=0}^{L-1}\mathbf{h}_k(\ell)\tilde{\mathbf{x}}_{k,l}(i-\ell),\nonumber\\
&\stackrel{(a)}{=}\sum_{\ell=0}^{L-1}\mathbf{h}_k(\ell){\mathbf{x}}_{k,l}(N+P-i-\ell),\nonumber\\
&=\tilde{\mathbf{r}}_l(N+P-i),\label{eq:row_similar}
\end{align}
where $\mathbf{x}_{k,l}$ is the $l$-th column of $\mathbf{X}_k$ and Step (a) follows using the structure of CP.
Now, using \eqref{eq:row_similar} and since the $i$-th row of $\mathbf{R}_{\rm N'}$ is $[\dots \tilde{\mathbf{r}}_l(i), \tilde{\mathbf{r}}_{l+1}(i)\dots]$, we can deduce that there are exactly $P-(L-1)$ duplicate rows in  $\mathbf{R}_{\rm N^\prime}$. This proves that the rank of $\mathbf{R}_{\rm N^\prime}$ is  equal to $N^\prime-(P-L+1)=N+L-1$ when $N^\prime=N+P$ and $N^\prime<M^\prime$. On the other hand, when $N^\prime\neq N+P$, the number of samples from previous block or symbol contributing interference to $\tilde{\mathbf{r}}_l$  varies as $l$ progresses. This is due to the incorrect segmentation will lead to spill over symbols unevenly across different columns of $\mathbf{R}_{\rm N^\prime}$. 
As a result, $\mathbf{R}_{\rm N^\prime}$ becomes full-rank.
\end{proof}
From Theorem \ref{thm:rank_property}, it can be observed that the selected $N^\prime$ can be classified as correct or incorrect based on the rank of segmented matrix $\mathbf{R}_{\rm N^\prime}$.We illustrate this rank-deficiency property using a small-scale example given in Appendix~\ref{app}.

It may be noted that the matrix $\mathbf{R}_{N^{\prime}}$ is no longer rank deficient when $\mathbf{r}$ received under noise. However, in this case, we use the eigenvalues of the corresponding covariance matrix to find the correct parameters. 
For the given $N'$, we determine the sample covariance matrix of $\mathbf{R}_{\rm N^\prime}$ as
\begin{equation}
\mathbf{C}_{\rm N^\prime} = \frac{1}{M'}\mathbf{R}_{\rm N^\prime}\mathbf{R}_{\rm N^\prime}^H.
\label{eq:Cr_def_ibi}
\end{equation}
Let $\lambda_{\rm N^\prime,1} \ge \cdots \ge \lambda_{\rm N^\prime,N^\prime}$ be the ordered eigenvalues of $\mathbf{C}_{\rm N^\prime}$ and let $N^\star=N+P$. Using Theorem \ref{thm:rank_property}, we have 
\begin{equation}
\mathrm{rk}(\mathbf{C}_{\rm N^\star})=N + L - 1, \label{eq:C_Nstar}
\end{equation}
for the noise-free case. Thus, for the noisy scenario, it is not difficult to see that the $P-(L-1)$ smallest eigenvalues will reduce to noise floor as the number of received OFDM symbols increases~\cite{MUSIC_1986}. This implies
\begin{equation}
    \lambda_{\rm N',\ell} \to \sigma_n^2 \text{~as~} M^*=\frac{KT}{N^\star}\to\infty,\label{eq:eigenvalue_sigma}
\end{equation}
for $\ell\in[N^\star+L-P-1:N^\star]$. However, for other $N^\prime\neq N^\star$,  $\lambda_{\rm N^\star,\ell} $ will converge to a value that is strictly above the noise floor with margin depending on the received SNR. 

From \eqref{eq:eigenvalue_sigma}, it is desirable to determine the optimal point of segregation as $p^*=N^\star+L-P-1$  when $N^\prime=N+P$ so that eigenvalue indices after and before $p^\star$ correspond to signal plus noise and only noise subspaces, respectively. This implies that $p^\star$ is the dimension of signal subspace. On the other hand, $p^\star$ should be equal to $N^\prime$ when $N^\prime\neq N+P$ as, in this case, all eigenvalues correspond to the signal plus noise subspace. For this purpose, minimum description length (MDL) criterion can be employed to accurately determine the dimension of signal subspace $p^\star$~\cite{MDL_Old_1989}. For given $\lambda_{\rm N^\prime,1}, \cdots, \lambda_{\rm N^\prime,N^\prime}$; the MDL criterion  evaluates   
\begin{align}
\mathrm{MDL}(\zeta;N^\prime) &= -(N^\prime-\zeta)M^\prime\log\left(\frac{\mathrm{GM}(\zeta;N')}{\mathrm{AM}(\zeta;N')}\right)\nonumber\\
&~~+\frac{1}{2}\zeta(2N^\prime-\zeta)\log M^\prime,\label{eq:MDL_def_ibi}
\end{align}

for $\zeta=1,\dots,N^\prime$ where 
\begin{align*}
\mathrm{GM}(\zeta;N^\prime) &=\Bigg(\prod_{i=\zeta+1}^{N^\prime} \lambda_{\rm N^\prime,i}\Bigg)^{\!\frac{1}{N^\prime-\zeta}},\\
\text{and~~}\mathrm{AM}(\zeta;N^\prime) & =\frac{1}{N'-\zeta}\sum_{i=\zeta+1}^{N^\prime} \lambda_{\rm N^\prime,i}.
\end{align*}
are the geometric  and arithmetic means of eigenvalues $\lambda_{\rm N^\prime,\zeta+1}, \cdots, \lambda_{\rm N^\prime,N^\prime}$. Recall, the AM--GM inequality, i.e. 
$\mathrm{GM}(\zeta;N') \le \mathrm{AM}(\zeta;N')$
with equality holding \emph{if and only if all residual eigenvalues are identical.}
This implies that $\mathrm{MDL}(\zeta;N^\star)$ attains its minimum at $\zeta=p^\star$  since $\mathrm{GM}(p^\star;N^\star)
\approx\mathrm{AM}(p^\star;N^\star)\approx\sigma_n^2$ (see \eqref{eq:eigenvalue_sigma}). 
Otherwise, for $N^\prime\neq N^\star$, $\mathrm{MDL}(\zeta;N^\prime)$ reaches its minimum  at $\zeta=N^\prime$, because $\mathrm{GM}(\zeta;N') < \mathrm{AM}(\zeta;N')$ for the monotonically decreasing eigenvalue sequence that does not exhibit convergence. To summarize this, we can write
\begin{align}
    \arg\min_\zeta\mathrm{MDL}(\zeta;N^\prime)=\begin{cases}
        p^\star &\text{if}~N^\prime =N^\star,\\
        N^\prime &\text{otherwise}.
    \end{cases}
\end{align}
The  minimum of the MDL curve appears distinctly only at
$N' = N^\star$, where the eigenvalue plateau of width $P-(L-1)$
reflects the cyclic-prefix redundancy preserved in the received symbol after interference spillover $L-1$ samples from the previous symbol/block. Using $p^\star=N^\star+L-P-1=N+L-1$, we can estimate the number of subcarriers as
\begin{equation}
    \hat{N}=  \arg\min_\zeta\mathrm{MDL}(\zeta;N^\star) - L+1.
\end{equation}
Since $p^\star=N^\prime + L-P-1=\arg\min_\zeta\mathrm{MDL}(\zeta;N^\prime)$ when $N^\prime=N^\star$ and $N^\star=N+P$, we can estimate $N$ as
\begin{equation}
     \hat{N}= \arg\min_{N^\prime} \Big|N^\prime +L-P-1 - \arg\min_\zeta\mathrm{MDL}(\zeta;N^\prime)\Big| - P.
 \end{equation}
Algorithm 1 presents the proposed approach for estimating the number of subcarriers $N$ using MDL criterion. 
\begin{algorithm}
\DontPrintSemicolon
\SetAlgoLined
\SetKwInOut{Input}{Input}
\SetKwInOut{Output}{Output}
\Input{ $\mathbf{r}$, $P$, $L$, $N_{\rm min}$, $N_{\rm max}$ }
\Output{$\hat{N}$.}
\BlankLine
\For{$N^\prime = N_{\rm min}: N_{\rm max}$}{
    Set $M^\prime = \lfloor KT / N^\prime \rfloor$\;
    \For{$k=1:M^\prime$}{
    $\tilde{\mathbf{r}}_k=[\mathbf{r}((k-1)N^\prime+1),\dots,\mathbf{r}(kN^\prime) ]^T$\;
    }
    \BlankLine
    Set $\mathbf{R}_{\rm N^\prime} = [\tilde{\mathbf{r}}_1, \tilde{\mathbf{r}}_2, \dots, \tilde{\mathbf{r}}_{M^\prime}]$\;
    Evaluate $\mathbf{C}_{\rm N^\prime} = \frac{1}{M^\prime} \mathbf{R}_{\rm N^\prime}\mathbf{R}_{\rm N^\prime}^{H}$\;
    Find eigenvalues of $\mathbf{C}_{\rm N^\prime}$ as $\lambda_{\rm N^\prime,1}, \cdots  \lambda_{\rm N^\prime,N^\prime}$\;
    Compute $\mathrm{MDL}(\zeta;N')$ using~\eqref{eq:MDL_def_ibi} for $\zeta = 1,\dots,N^\prime$\;
    $\hat{\zeta}(N') = \arg\min\limits_\zeta \mathrm{MDL}(\zeta;N')$\;
}
\BlankLine
$\hat{N} = \arg\min\limits_{N'} \big|\,\hat{\zeta}(N^\prime) - (N^\prime + L - 1 - P)\,\big|-P$\;
\Return $\hat{N} $\;

\caption{Estimating the number of subcarriers via eigen–spectrum MDL criterion}
\label{alg:subcarrier_detection}
\end{algorithm}

\section{Numerical Results and Discussion}
\label{sec:results}
This section presents the numerical analysis of  the proposed method for estimating the number of subcarriers of CP-OFDM. In particular, we numerically characterize the performance  in terms of detection probability (${\rm P_d}$) 
through Monte Carlo simulation with  1000 iterations, conducted over a wide range of system parameters to highlight the robustness of the proposed method. For numerical results, we consider number of subcarriers $N=64$, CP length $P=7$, number of delay taps $L=6$, number of blocks $K=5$, number of symbols per block $M=500$; unless mentioned otherwise. 

\begin{figure}[h]
         \centering\includegraphics[width=0.45\textwidth]{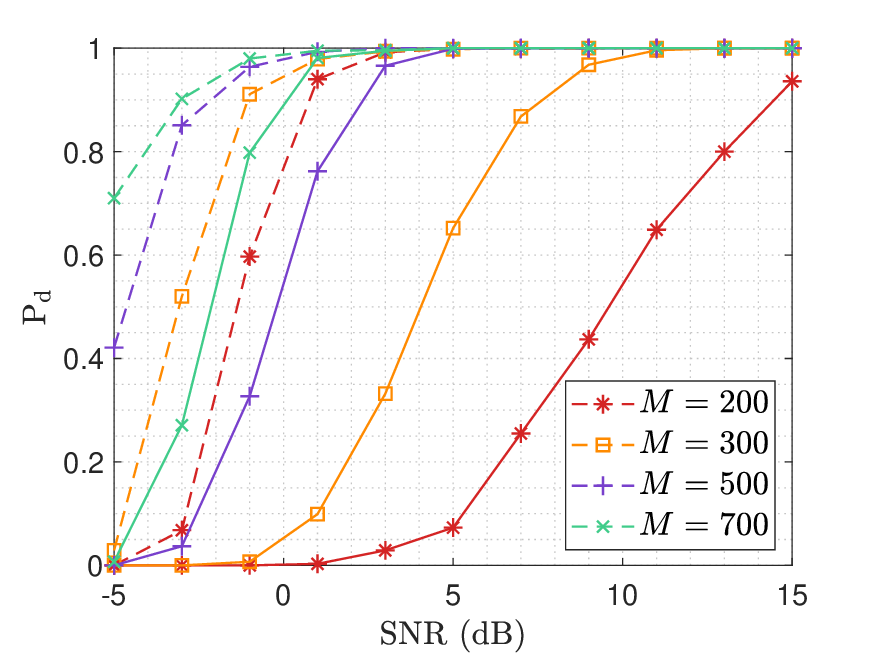}
        \caption{${\rm P_d}$ vs. SNR. The solid and dashed curves represent ${\rm P_d}$ for $N=64$ and $N=32$, respectively. }\vspace{-4mm}
        \label{fig:Pd_SNR}
\end{figure}
Fig. \ref{fig:Pd_SNR} shows that the  probability ${\rm P_d}$ of correct detection of the number of subcarriers $N$ for wide range of signal-to-noise ratio (SNR). It can be seen from the figure that the proposed method provides reasonably high ${\rm P_d}$ even at low SNR, particularly for smaller number of subcarriers $N$. Further, it can be observed that ${\rm P_d}$ improves with the increase in number of observation symbols $M$, however it degrades as the number of subcarriers $N$ increases. 

\begin{figure}[h]
         \centering\includegraphics[width=0.45\textwidth]{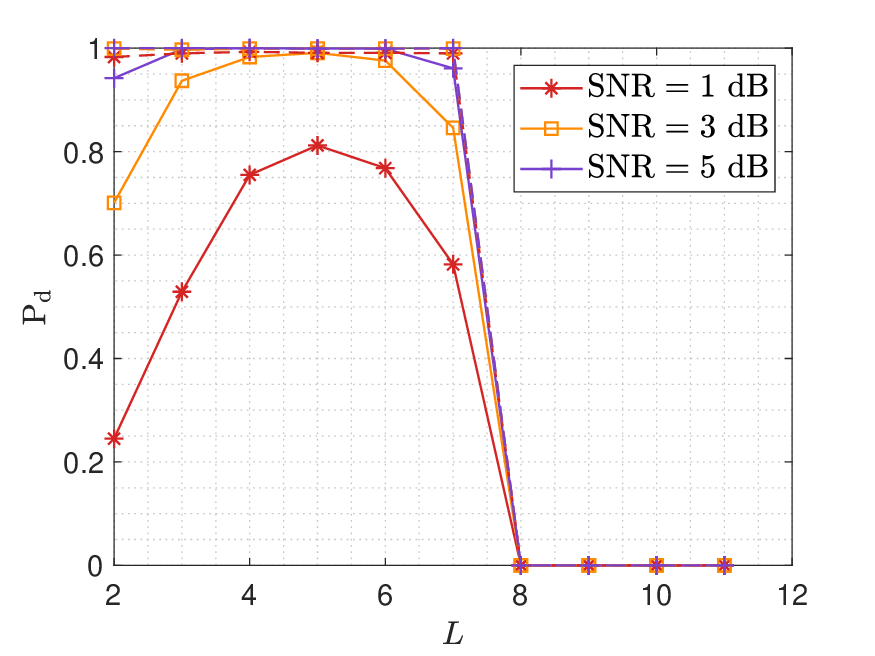}\vspace{-2mm}
        \caption{${\rm P_d}$ vs. $L$. The solid and dashed curves represent ${\rm P_d}$ for $N=64$ and $N=32$, respectively. }\vspace{-2mm}
        \label{fig:Pd_L}
\end{figure}
Fig. \ref{fig:Pd_L} shows the behavior of detection probability ${\rm P_d}$ as a function of the number of delay taps for CP $P=7$. It can be observed that  ${\rm P_d}$ is higher for $L\leq P$ and drops to zero for $L>P$. This is expected since the rank deficiency  property holds only for the case of $L\leq P$ (as shown in Theorem \ref{thm:rank_property}), and thus the proposed algorithm breaks down if this condition is not satisfied. However, it should be noted that the OFDM inherently implies that CP length should be greater than the delay spread to mitigate the inter-symbol interference, which ensures that condition required for the proposed method will always be satisfied. Besides, the figure also shows that the ${\rm P_d}$  peaks for a particular $L$. This may be due to the better estimation of eigen-spectrum breakpoints using MDL separating the signal and noise subspaces occurs at $L$ is smaller than $P$ but not close to zero.  
However, the overall performance becomes robust as the {\rm SNR} increase and/or $N$ decreases.

\begin{figure}[h]
         \centering\includegraphics[width=0.45\textwidth]{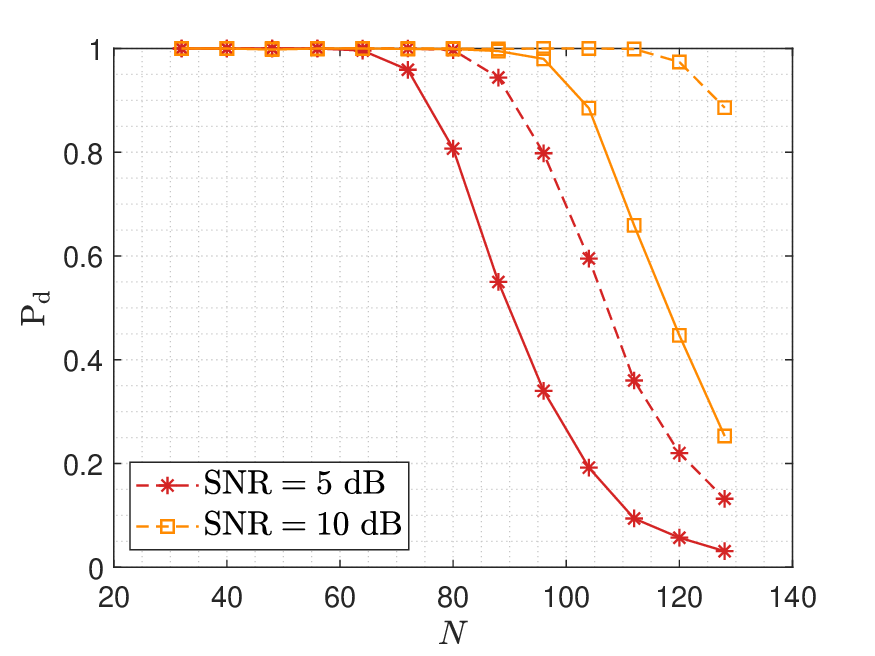}\vspace{-3mm}
        \caption{${\rm P_d}$ vs. $N$. The solid and dashed curves represent ${\rm P_d}$ for $P=7$ and $P=10$, respectively. }\vspace{-4mm}
        \label{fig:Pd_N}
\end{figure}
\begin{figure}[h]
         \centering\includegraphics[width=0.45\textwidth]{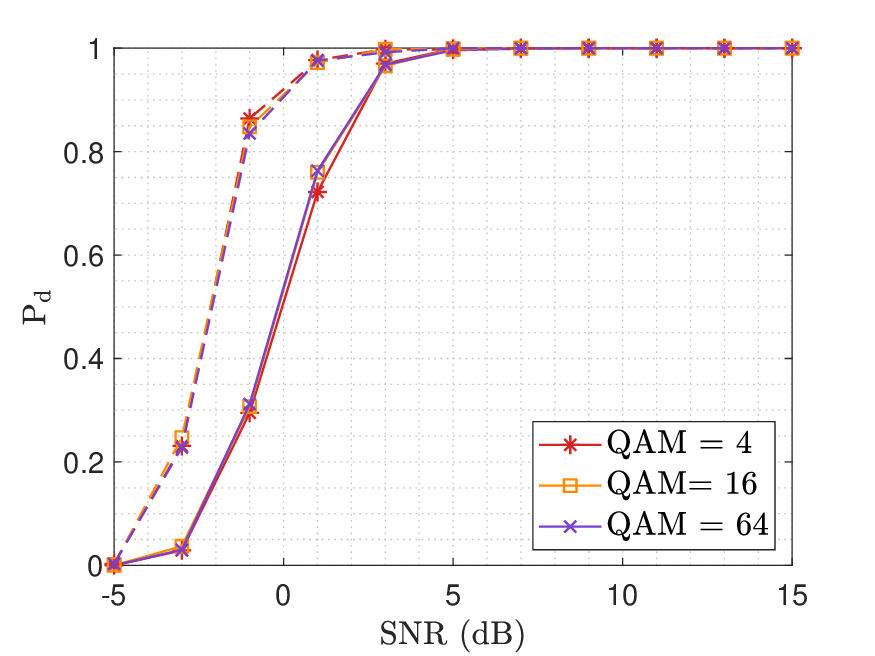}
        \caption{${\rm P_d}$ vs. {\rm SNR}. The solid and dashed curves represent ${\rm P_d}$ for $P=7$ and $P=10$, respectively. }\vspace{-2mm}
        \label{fig:Pd_SNR_QAM}
\end{figure}
\begin{figure*}[t!]
\setlength{\arraycolsep}{1pt}

\begin{equation}
\footnotesize
    \mathbf{R}_{4}=
\begin{bmatrix}
\mathbf{h}_1(1) \tilde{\mathbf{X}}_1(1,1) &
\mathbf{h}_1(1)\tilde{\mathbf{X}}_1(2,1)+\mathbf{h}_1(2)\tilde{\mathbf{X}}_1(1,2) &
\mathbf{h}_2(1)\tilde{\mathbf{X}}_2(1,1)+\mathbf{h}_2(2)\tilde{\mathbf{X}}_1(2,2) &
\mathbf{h}_2(1)\tilde{\mathbf{X}}_2(2,1)+\mathbf{h}_2(2)\tilde{\mathbf{X}}_2(1,2)\\

\cellcolor{gray!15}\mathbf{h}_1(1)\tilde{\mathbf{X}}_1(1,2)+\mathbf{h}_1(2)\tilde{\mathbf{X}}_1(1,1) &
\cellcolor{gray!15}\mathbf{h}_1(1)\tilde{\mathbf{X}}_1(2,2)+\mathbf{h}_1(2)\tilde{\mathbf{X}}_1(2,1) &
\cellcolor{gray!15}\mathbf{h}_2(1)\tilde{\mathbf{X}}_2(1,2)+\mathbf{h}_2(2)\tilde{\mathbf{X}}_2(1,1) &
\cellcolor{gray!15}\mathbf{h}_2(1)\tilde{\mathbf{X}}_2(2,2)+\mathbf{h}_2(2)\tilde{\mathbf{X}}_2(2,1)\\
\mathbf{h}_1(1)\tilde{\mathbf{X}}_1(1,1)+\mathbf{h}_1(2)\tilde{\mathbf{X}}_1(1,2) &
\mathbf{h}_1(1)\tilde{\mathbf{X}}_1(2,1)+\mathbf{h}_1(2)\tilde{\mathbf{X}}_1(2,2) &
\mathbf{h}_2(1)\tilde{\mathbf{X}}_2(1,1)+\mathbf{h}_2(2)\tilde{\mathbf{X}}_2(1,2) &
\mathbf{h}_2(1)\tilde{\mathbf{X}}_2(2,1)+\mathbf{h}_2(2)\tilde{\mathbf{X}}_2(2,2)\\
\cellcolor{gray!15}\mathbf{h}_1(1)\tilde{\mathbf{X}}_1(1,2)+\mathbf{h}_1(2)\tilde{\mathbf{X}}_1(1,1) &
\cellcolor{gray!15}\mathbf{h}_1(1)\tilde{\mathbf{X}}_1(2,2)+\mathbf{h}_1(2)\tilde{\mathbf{X}}_1(2,1) &
\cellcolor{gray!15}\mathbf{h}_2(1)\tilde{\mathbf{X}}_2(1,2)+\mathbf{h}_2(2)\tilde{\mathbf{X}}_2(1,1) &
\cellcolor{gray!15}\mathbf{h}_2(1)\tilde{\mathbf{X}}_2(2,2)+\mathbf{h}_2(2)\tilde{\mathbf{X}}_2(2,1)
\end{bmatrix}
\tag{14}
\label{eq:R4_example}
\end{equation}
\hrulefill
\end{figure*}
Fig. \ref{fig:Pd_N} shows that the ${\rm P_d}$ degrades with the increase in the number of subcarriers $N$, which is expected. However, it decreases slowly with increasing of  {\rm SNR} and/or $P$, ensuring that the proposed method is robust against $N$ at larger {\rm SNR}. The figure also shows that the proposed method works for arbitrary values of $N$.

Fig. \ref{fig:Pd_SNR_QAM} shows that the ${\rm P_d}$ performance is independent of the order of modulation scheme and improves rapidly  with increase in {\rm SNR} and the length of CP $P$. The independence of modulation order differentiates the proposed subspace based approach from the statistical properties based approaches available in the literature, whose performances are usually sensitive to the modulation schemes. From Fig. \ref{fig:Pd_N} and Fig. \ref{fig:Pd_SNR_QAM},  it can be observed that higher CP length $P$ provides better detection performance at low SNR and higher number of subcarriers $N$, especially when it is moderately higher than the number of delay taps $L$ as can be verified from Fig. \ref{fig:Pd_L}







\vspace{-2mm}
\section{Conclusion}
\label{Section_conclusion}

In this work, we studied the estimation problem  of estimating the number of subcarriers in an OFDM system under the multipath block fading scenario. For this problem, we proposed a novel method based on the eigen-spectral analysis of the covariance matrix corresponding to the received data. 
We also studied the rank properties of this covariance matrix, in particular we showed that the covariance matrix has  rank property for the correct segmentation of the received symbols  that is distinct from other segmentation  under quasi block fading. This fact facilitates the accurate estimation of the number of subcarriers.
The performance of our method was validated via various simulation results for wide range of parameters. The numerical results show that the proposed method offers high detection probability at low SNR. Further, unlike the existing methods in the literature,  the proposed  method works for an arbitrary number of subcarriers while ensuring performance robustness to the order of modulation scheme.


\vspace{-3mm}
\appendix
\section*{Example of Rank Deficiency of $\mathbf{R}_{N'}$}
\label{app}
Here, we provide an example to verify the rank-reduction property of $\mathbf{R}_{\rm N^\prime}$ as mentioned in  \cref{thm:rank_property}, 
Consider a small-scale setup with $N=2, P=2, L=2, M=2, \text{~and~} K=2$. Here, note that $P \ge L-1$, and $T = M(N+P) = 8$. For $L=2$, intra- and inter-block convolution matrices become
\begin{align*}
\mathbf{H}_k(i,j)&=\begin{cases}
    \mathbf{h}_k(1), &\text{if~} i=j,\\
    \mathbf{h}_k(2), &\text{if~} i-j=1,\\
    0, &\text{else},
\end{cases}\\
 \mathbf{B}_1(i,j)&=0 \text{~~for~} i,j, \text{~and}\\
\mathbf{B}_2(i,j)&=\begin{cases}
    \mathbf{h}_2(2), &\text{if~} i=1,j=12,\\
    0, &\text{else}.
\end{cases}
\end{align*}
Thus, for noise free case, received sequence blocks can be modeled using \eqref{eq:block_model} as
\begin{align*}
    \mathbf{r}_1 =\mathbf{H}_1\mathbf{s}_1=
&\begin{bmatrix}
\mathbf{h}_1(1) \tilde{\mathbf{X}}_1(1,1)\\
\mathbf{h}_1(1)\tilde{\mathbf{X}}_1(1,2)+\mathbf{h}_1(2)\tilde{\mathbf{X}}_1(1,1)\\
\mathbf{h}_1(1)\tilde{\mathbf{X}}_1(1,1)+\mathbf{h}_1(2)\tilde{\mathbf{X}}_1(1,2)\\
\mathbf{h}_1(1)\tilde{\mathbf{X}}_1(1,2)+\mathbf{h}_1(2)\tilde{\mathbf{X}}_1(1,1)\\
\mathbf{h}_1(1)\tilde{\mathbf{X}}_1(2,1)+\mathbf{h}_1(2)\tilde{\mathbf{X}}_1(1,2)\\
\mathbf{h}_1(1)\tilde{\mathbf{X}}_1(2,2)+\mathbf{h}_1(2)\tilde{\mathbf{X}}_1(2,1)\\
\mathbf{h}_1(1)\tilde{\mathbf{X}}_1(2,1)+\mathbf{h}_1(2)\tilde{\mathbf{X}}_1(2,2)\\
\mathbf{h}_1(1)\tilde{\mathbf{X}}_1(2,2)+\mathbf{h}_1(2)\tilde{\mathbf{X}}_1(2,1)
\end{bmatrix},\\
\mathbf{r}_2 =\mathbf{H}_2\mathbf{s}_2 + \mathbf{B}_2\mathbf{s}_1=
&\begin{bmatrix}
\mathbf{h}_2(1)\tilde{\mathbf{X}}_2(1,1)+\mathbf{h}_2(2)\tilde{\mathbf{X}}_1(2,2)\\
\mathbf{h}_2(1)\tilde{\mathbf{X}}_2(1,2)+\mathbf{h}_2(2)\tilde{\mathbf{X}}_2(1,1)\\
\mathbf{h}_2(1)\tilde{\mathbf{X}}_2(1,1)+\mathbf{h}_2(2)\tilde{\mathbf{X}}_2(1,2)\\
\mathbf{h}_2(1)\tilde{\mathbf{X}}_2(1,2)+\mathbf{h}_2(2)\tilde{\mathbf{X}}_2(1,1)\\
\mathbf{h}_2(1)\tilde{\mathbf{X}}_2(2,1)+\mathbf{h}_2(2)\tilde{\mathbf{X}}_2(1,2)\\
\mathbf{h}_2(1)\tilde{\mathbf{X}}_2(2,2)+\mathbf{h}_2(2)\tilde{\mathbf{X}}_2(2,1)\\
\mathbf{h}_2(1)\tilde{\mathbf{X}}_2(2,1)+\mathbf{h}_2(2)\tilde{\mathbf{X}}_2(2,2)\\
\mathbf{h}_2(1)\tilde{\mathbf{X}}_2(2,2)+\mathbf{h}_2(2)\tilde{\mathbf{X}}_2(2,1)
\end{bmatrix}.
\end{align*}
Concatenating $\mathbf{r}_1$ and $\mathbf{r}_2$  gives the  received sequence  of length $KT=16$ as $\mathbf{r}=[\mathbf{r}_1^T,\mathbf{r}_2^T]^T$. Now, segmenting $\mathbf{r}$ with $N^\prime=N+P=4$, gives a $4\times 4$ matrix as given in \eqref{eq:R4_example}. From \eqref{eq:R4_example}, it can be observed that the second and fourth rows (highlighted in gray) are identical, i.e.  $\mathbf{R}_4(2,j)=\mathbf{R}_4(4,j)$ for $j=1,\dots,4$. 
This deterministic linear dependence implies that 
\begin{equation*}
\mathrm{rk}(\mathbf{R}_4) = N+L- 1 = 3.    
\end{equation*}
Similarly, for other choices of $N^\prime\neq 4$, the matrix $\mathbf{R}_{\rm N^\prime}$ can be constructed to verify that $\mathrm{rk}(\mathbf{R}_N^\prime)=N^\prime$.

\vspace{-3mm}
\bibliographystyle{IEEEtran}


\end{document}